\documentclass[11pt]{article}
\usepackage{amsmath, amsthm}
\usepackage{latexsym}
\usepackage{graphicx}
\makeatletter
\newfont{\footsc}{cmcsc10 at 8truept}
\newfont{\footbf}{cmbx10 at 8truept}
\newfont{\footrm}{cmr10 at 10truept}
\makeatother
\newtheorem{theorem}{\bf Theorem}

\begin{document}
\title{D-optimal designs via a cocktail algorithm}

\author{Yaming Yu\\
\small Department of Statistics\\[-0.8ex]
\small University of California\\[-0.8ex] 
\small Irvine, CA 92697, USA\\[-0.8ex]
\small \texttt{yamingy@uci.edu}}

\date{}
\maketitle

\begin{abstract}
A fast new algorithm is proposed for numerical computation of (approximate) D-optimal designs.  This {\it cocktail algorithm} extends the well-known vertex direction method (VDM; Fedorov 1972) and the multiplicative algorithm (Silvey, Titterington and Torsney, 1978), and shares their simplicity and monotonic convergence properties.  Numerical examples show that the cocktail algorithm can lead to dramatically improved speed, sometimes by orders of magnitude, relative to either the multiplicative algorithm or the vertex exchange method (a variant of VDM).  Key to the improved speed is a new nearest neighbor exchange strategy, which acts locally and complements the global effect of the multiplicative algorithm.  Possible extensions to related problems such as nonparametric maximum likelihood estimation are mentioned. 

{\bf Keywords:} D-optimality; experimental design; hybrid algorithm. 
\end{abstract}

\section{Introduction}
This paper studies numerical methods for computing D-optimal designs (approximate theory; see Kiefer 1974, Pukelsheim 1993, and Atkinson, Donev and Tobias, 2007).  Given a parametric model, the problem is to find an allocation of weights to the design points $x_1, \ldots, x_n$ (which encode the explanatory variables at specific values) so that the determinant of the Fisher information matrix of the parameter is maximized.  We focus on the linear model and discuss possible extensions in Section~5.  Two strategies for this classical problem are the vertex direction method (VDM; see Fedorov 1972 and Wynn 1972) and the multiplicative algorithm (Silvey, Titterington and Torsney, 1978).  Both VDM and the multiplicative algorithm are simple iterative strategies that converge monotonically, i.e., the determinant criterion never decreases along the iterations.  Though easy to implement, VDM or the multiplicative algorithm can be slow, and various strategies have been devised to remedy this.  In particular, B\"{o}hning (1986) proposes the vertex exchange method (VEM) as a more effective variant of VDM.  Variants of the multiplicative algorithm are considered by, for example, Titterington (1978), Mandal and Torsney (2006), and Dette, Pepelyshev and Zhigljavsky (2008). 

In this work we propose a {\it cocktail algorithm} for efficient computation of D-optimal designs.  As the name suggests, this is based on a combination of several strategies, including VDM and the multiplicative algorithm.  A new ingredient that contributes significantly to its effectiveness, however, is a nearest neighbor exchange strategy, which is intended to complement the multiplicative algorithm.  Two desirable effects of nearest neighbor exchanges are i) elimination of multiple bad support points, and ii) quick apportionment between very similar support points.  Both compensate for the potentially slow convergence rate of the multiplicative algorithm, while the former also reduces its computing time per iteration.  Operationally, the nearest neighbor exchanges are as simple to implement as VDM, VEM, or the multiplicative algorithm.  The speedup brought in by such a simple modification, however, can be dramatic. 

In Section~2, after a brief review of the D-optimal design problem on finite design spaces, we describe the multiplicative algorithm, VDM, and VEM.  Then we introduce the nearest neighbor exchange strategy and formally define the cocktail algorithm.  Section~3 establishes that the cocktail algorithm is monotonically convergent.  Section~4 presents numerical illustrations with several regression models.  The cocktail algorithm compares favorably with the multiplicative algorithm, VEM, and general optimization methods such as conjugate gradient and quasi-Newton.  Section~5 concludes with a discussion on possible extensions. 

\section{Algorithms for D-optimal designs}
We focus on the important case of a finite design space $\mathcal{X}=\{x_1, \ldots, x_n\}\subset \mathbf{R}^m$, 
which may be the result of discretizing an underlying continuous space.  An {\it approximate design} (Kiefer, 1974) 
is any probability vector $w=(w_1,\ldots, w_n)\in \bar{\Omega}$, where $\bar{\Omega}$ denotes the closure of $\Omega=\{w:\ \sum_{i=1}^n w_i=1,\ w_i> 0\}$.  The value $w_i$ represents the proportion of units an experimenter assigns to $x_i$.  Approximate designs allow $w_i$ to be real numbers; some rounding is usually used to convert $w$ to a design with a finite sample size.  Suppose the response from a unit assigned to $x_i$ is modeled as 
$$y|(x_i,\theta)\sim {\rm N}(x_i^\top \theta, \sigma^2),$$
where $\theta$ ($m\times 1$) is the parameter of interest, and suppose responses from different units are independent.  Then the Fisher information matrix for $\theta$ is proportional to 
$$M(w)=\sum_{i=1}^n w_ix_ix_i^\top.$$
A design $w^*$ is D-optimal if it maximizes 
$$\phi(w)\equiv \log \det M(w),\quad w\in \bar{\Omega}.$$  
Equivalently, a D-optimal design minimizes the determinant of the variance matrix of the best linear unbiased estimator of $\theta$.  The D-criterion is among the most widely used optimal design criteria. 

We shall describe several iterative algorithms for finding D-optimal designs.  These differ in the choice of the starting value $w^{(0)}$ and the updating rule $w^{(t)}\to w^{(t+1)}$.  The following common convergence criterion,  however, will be used throughout.  Define 
$$d(i, j, w)\equiv x_i^\top M^{-1}(w) x_j,\quad d(i, w)\equiv d(i, i, w).$$
Note that $d(i, w)=\partial \phi(w)/\partial w_i $.  Alternatively, $d(i, w)-m$ is a directional derivative $\partial\phi((1-\delta) w+\delta e_i)/\partial \delta |_{\delta=0+}$ where the probability vector $e_i$ puts all the mass on $x_i$. 

{\bf Convergence criterion:}
\begin{equation}
\label{conv}
m^{-1}\max_{1\leq i\leq n} d\left(i, w^{(t)}\right) \leq 1+\epsilon,
\end{equation}
where $\epsilon$ is a small positive constant. 

This convergence criterion can be motivated from the general equivalence theorem (Kiefer and Wolfowitz, 1960), part of which states that $w$ is D-optimal if and only if 
$$m^{-1} \max_{1\leq i\leq n} d(i, w)=1.$$ 
The general equivalence theorem thus allows us to check whether a given weight allocation $w=(w_1,\ldots, w_n)$ is D-optimal; it is crucial to both analytic and numerical approaches to the problem. 

\subsection{The multiplicative algorithm}
The {\it multiplicative algorithm} (MA) refers to a well-known proposal of Silvey et al. (1978). 

{\bf Algorithm I (the multiplicative algorithm)}

{\it Starting value.} Choose $w^{(0)}\in \Omega$, i.e., $w_i^{(0)}>0$ for all $i$.  

{\it Updating rule.} 
\begin{equation}
\label{alg1}
w_i^{(t+1)}=w_i^{(t)} m^{-1} d\left(i, w^{(t)}\right),\quad i=1, \ldots, n.
\end{equation}

Let us denote the mapping (\ref{alg1}) as $w^{(t+1)}=MA(w^{(t)})$.  Equivalently, (\ref{alg1}) can be written as
\begin{equation}
\label{ma2}
w_i^{(t+1)}=w_i^{(t)}\frac{\partial \phi\left(w^{(t)}\right)/\partial w_i}{\sum_{j=1}^n w_j^{(t)} \partial 
\phi\left(w^{(t)}\right)/\partial w_j},
\end{equation}
which highlights $\sum_i w_i^{(t+1)}=1$, that is, $w^{(t+1)}$ is correctly normalized.  Heuristically, (\ref{ma2})  simply adjusts the weights $w$ so that proportionally more weight is put on $x_i$ if 
the gain in the objective function $\phi$ by a slight increase in $w_i$ (i.e., $\partial \phi(w)/\partial w_i$) is larger. 

Algorithm~I has generated considerable interest; see, for example, Titterington (1976, 1978), Silvey et al. (1978), Mandal and Torsney (2006), Harman and Pronzato (2007), and Dette et al.\ (2008).  The latter three papers are concerned with improving the multiplicative algorithm based on principles different from the exchange strategies reported here.  Mandal and Torsney (2006) consider applying a class of multiplicative algorithms to clusters of design points for better efficiency.  Harman and Pronzato (2007) study methods to exclude nonoptimal design points so that the dimension of the problem is reduced.  Dette et al.\ (2008) propose a modification of Algorithm~I which takes larger steps at each iteration but still maintains monotonic convergence (see also Yu 2010b).  Another relevant work is Yu (2010a), which formulates Algorithm~I as an iterative conditional minimization procedure and is mainly concerned with theoretical properties.  Algorithm~I is an important ingredient in our proposed cocktail algorithm. 

\subsection{VDM and VEM}
The vertex direction method (VDM) is defined by the following iteration $w\to w^{new}$. 

{\bf VDM:} 
Select $1\leq i_{max}\leq n$ such that 
\begin{equation}
\label{VDM}
d(i_{max}, w)=\max_{1\leq i\leq n} d(i, w),
\end{equation}
and set $w^{new}=VDM(w)$ as 
\begin{equation*}
w^{new}_i=\begin{cases} (1-\delta) w_i,& i\neq i_{max},\\
(1-\delta) w_i+\delta, & i=i_{max},
\end{cases}
\end{equation*}
where $\delta\in [0,1]$ is such that $\det M(w^{new})$ is maximized.  The maximizing $\delta$ is available in closed form: 
$$\delta=\frac{d\left(i_{max}, w\right)/m-1}{d\left(i_{max}, w\right)-1}.$$
See Fedorov (1972) for the underlying rationale.  Plainly, we move $w$ in the direction of a design point toward which the directional derivative of $\phi$ is the greatest.  VDM is a steepest ascent strategy in this sense.  With a slight abuse of notation, we shall occasionally write $w^{new}=VDM(i_{max}, w)$ instead of $w^{new}=VDM(w)$ to emphasize the index $i_{max}$. 

Closely related to VDM is a general exchange step $w\to w^{new}$ for any two design points $x_j, x_k,\ j\neq k$. 

{\bf VE(j, k):}
Set $w^{new}$ as 
\begin{equation*}
w^{new}_i =\begin{cases} w_i,& i\notin\{j, k\},\\
w_i-\delta, & i=j,\\
w_i+\delta, & i=k,
\end{cases}
\end{equation*}
where $\delta\in \left[-w_k, w_j\right]$ is chosen such that $\det M(w^{new})$ is maximized.  Following B\"{o}hning (1986), it can be shown that the maximizing $\delta$ is 
$$\delta=\min\left\{w_j,\, \max\{-w_k,\, \delta^*(j, k)\}\right\},$$
where
\begin{equation}
\label{dstar}
\delta^*(j, k)=\frac{d(k, w)-d(j, w)}{2(d(j, w)d(k, w)-d^2(j, k, w))}.
\end{equation}

Plainly, VE($j, k$) performs an optimal exchange of mass between $x_j$ and $x_k$.  The exchange is optimal in the sense of maximal increase in the determinant of the information matrix.  When $\delta=w_j$, all the mass assigned to $x_j$ (which has a smaller $d(j, w)$, indicating that it should carry less weight) is transferred to $x_k$; similarly when $\delta=-w_k$.  We shall denote this mapping $w\to w^{new}$ by $w^{new}=VE(j, k, w)$. 

{\bf Remark.} The denominator in (\ref{dstar}) is nonnegative by Cauchy-Schwarz.  It becomes zero only when one of $x_j, x_k$ is a constant multiple of the other, in which case we define $\delta^*(j, k)$ as $+\infty$ or $-\infty$ according as $d(k, w)>d(j, w)$ or $d(k, w)<d(j, w)$.  If $x_j+x_k=0$, then both the numerator and the denominator in (\ref{dstar}) become zero, and $\det M\left(w^{new}\right)$ is constant as a function of $\delta\in [-w_k, w_j]$; we set $\delta^*(j, k)$ as an arbitrary constant (say zero) in this case.  These contingencies rarely arise in practice. 

The {\it vertex exchange} method (VEM) of B\"{o}hning (1986) performs an optimal exchange between two special design points. 

{\bf Algorithm II (the vertex exchange method)} 

{\it Starting value.}
Choose $w^{(0)}\in \bar{\Omega}$ such that $\det M(w^{(0)})>0$. 

{\it Updating rule.}
Select $i_{min}$ and $i_{max}$ such that 
\begin{align}
\label{VEM1}
d\left(i_{min}, w^{(t)}\right)&=\min\left\{d(i, w^{(t)}):\ w^{(t)}_i>0\right\},\\
\label{VEM2}
d\left(i_{max}, w^{(t)}\right)&=\max_{1\leq i\leq n} d\left(i, w^{(t)}\right).
\end{align}
Set 
$$w^{(t+1)}=VE\left(i_{max}, i_{min}, w^{(t)}\right).$$ 

That is, VEM finds $i_{min}$ (resp.\ $i_{max}$) such that $d(i, w)$ is minimized (resp.\ maximized), and then performs an optimal transfer of mass from $x_{i_{min}}$ to $x_{i_{max}}$ (it is also required that $i_{min}$ have nonzero mass to supply to $i_{max}$).  Hence we may view VEM as a steepest ascent strategy in its choice of the two indices $i_{min}$ and $i_{max}$.  We shall numerically compare VEM with our proposed algorithm in Section~4. 

\subsection{Nearest neighbor exchanges}
A key ingredient in our proposed algorithm is a nearest neighbor exchange strategy, which can be motivated as follows.  Intuitively, Algorithm I (the multiplicative algorithm) may have difficulty apportioning the mass between adjacent design points.  Consider two design points $x_i$ and $x_j$ that are close together as measured by some distance metric in $\mathbf{R}^m$.  Then $d(i, w)\approx d(j, w)$, and according to (\ref{alg1}), we have
$$\frac{w_i^{(t+1)}}{w_j^{(t+1)}}=\frac{w_i^{(t)}}{w_j^{(t)}}\frac{d\left(i, w^{(t)}\right)}{d\left(j, w^{(t)}\right)} \approx \frac{w_i^{(t)}}{w_j^{(t)}}.$$
That is, the relative proportions between $x_i$ and $x_j$ barely change from iteration to iteration.  Another way of putting it is that, if $x_i$ is a support point of the optimal design, then it would take many iterations before Algorithm I can significantly reduce the mass on those $x_j$ which are adjacent to $x_i$ but are not support points. 

A simple remedy is to add nearest neighbor exchanges (NNEs) to Algorithm I.  NNEs are easy to define when there exists a natural ordering in the design space.  An example is 
\begin{equation}
\label{space1}
\mathcal{X}=\left\{x_i=(1,\, f(i/n))^\top,\ i=1,\ldots, n\right\},
\end{equation}
where $f$ is a continuous function on $[0,1]$ representing a single quantitative predictor.  In such a case $x_i$ and $x_j$ are close whenever $|i-j|$ is small.  Given the current iterate $w^{(t)}$, let $i_1<\cdots<i_{p+1}$ denote the indices of the support points of $w^{(t)}$.  We may consider performing 
vertex exchanges between $x_{i_j}$ and $x_{i_{j+1}}$ for $j=1,\ldots, p$ in turn, i.e., 
$$w^{(t+j/p)}=VE\left(i_j, i_{j+1}, w^{(t+(j-1)/p)}\right),\quad j=1,\ldots, p,$$
where fractional superscripts denote intermediate output.  We refer to the mapping $w^{(t)}\to w^{(t+1)},$
which consists of $p$ sub-steps, as the set of nearest neighbor exchanges.  Note that non-support points of $w^{(t)}$ are excluded, i.e., $x_{i_{j+1}}$ is a ``nearest neighbor'' of $x_{i_j}$ in the support of $w^{(t)}$ only. 

This intuitively appealing prescription depends on a natural ordering of $x_i$.  Sometimes there is no single natural ordering, e.g., when the design space encodes two or more factors.  Selecting an ordering that best captures the neighborhood structure is therefore an interesting problem.  In our numerical examples (Section~4), we explore another approach, which dynamically determines the nearest neighbors at each iteration.  Specifically, let $\|x_j-x_k\|$ denote the distance between design points $x_j$ and $x_k$, as measured by the $L_1$ norm.  The choice of the metric does not make much difference in our experience. 

{\bf NNE:} Let $i_1,\ldots,i_{p+1}$ be the elements of $\{i:\ w_i^{(t)}>0\}$ where $p+1$ is the number of support points of $w^{(t)}$.  For each $i_j,\ j=1,\ldots, p,$ let $i^*_j$ be any index $i\in \{i_{j+1},\ldots, i_{p+1}\}$ such that $\left\|x_i -x_{i_j}\right\|$ is minimized.  Perform vertex exchanges between $x_{i_j}$ and $x_{i_j^*}$ for $j=1,\ldots, p$ in turn, i.e., 
$$w^{(t+j/p)}=VE\left(i_j, i_j^*, w^{(t+(j-1)/p)}\right).$$

Again, non-support points of $w^{(t)}$ are excluded.  We shall denote the composite mapping $w^{(t)}\to w^{(t+1)}$ as $w^{(t+1)}=NNE(w^{(t)})$. 

{\bf Remark.} The index $i_j^*$ is defined as a minimizer of $\left\|x_i-x_{i_j}\right\|$ over $i\in \{i_{j+1},\ldots, i_{p+1}\}$, rather than over $i\in \{i_1, \ldots, i_{j-1}, i_{j+1}, \ldots, i_{p+1}\}$, to avoid possible redundancies.  If we adopt the latter definition, then for two points that are nearest neighbors of each other, we would have two  exchange steps in one iteration between these same points. 

\subsection{The cocktail algorithm}
NNE has a serious problem as a stand-alone algorithm.  By definition, we have $w_i^{(t+1)}=0$ once $w_i^{(t)}=0$, 
i.e., the point $x_i$ remains outside of the support set.  Algorithm I, which suffers from the same problem, circumvents it by assigning positive initial mass to each design point.  NNE may result in $w^{(t+1)}_i=0$ even if $w^{(t)}_i>0$.  The problem persists when we combine NNE and Algorithm I. 

An easy solution is to add in the updating rule of VDM.  By definition, a VDM step can put some mass on a design point that was assigned zero mass previously.  We define {\it the cocktail algorithm} as a combination of VDM, NNE, and Algorithm I. 

{\bf Algorithm III (the cocktail algorithm)}

{\it Starting value.}
Choose $w^{(0)}\in \bar{\Omega}$ such that $\det M(w^{(0)})>0$. 

{\it Updating rule.} Perform an iteration of VDM, the nearest neighbor exchanges, and then an iteration of Algorithm I.
That is, let 
\begin{equation}
\label{cocktail}
w^{(t+1/3)}=VDM(w^{(t)}),\quad w^{(t+2/3)}=NNE(w^{(t+1/3)}),\quad w^{(t+1)}=MA(w^{(t+2/3)}),
\end{equation}
where again fractional superscripts indicate intermediate output. 

An added benefit of (\ref{cocktail}) is that NNE helps keep the number of support points of $w^{(t+2/3)}$ small, so that each iteration of $w^{(t+1)}=MA(w^{(t+2/3)})$ costs little time, as we need not update the coordinates of $w^{(t+2/3)}$ that are zero at the multiplicative step. 

We may consider using a VEM step instead of the VDM step above.  The resulting algorithm is similarly effective (numerical comparison omitted), although the convergence proof of Section~3 does not seem to extend easily to this alternative cocktail algorithm. 

\section{Monotonic convergence}
Several algorithms considered here have an appealing monotonic convergence property, i.e., as $t\uparrow \infty,\ \det M(w^{(t)})$ increases to $\sup_{w\in \Omega} \det M(w)$.  Monotonic convergence of Algorithm I is well-established (see Titterington, 1976; P\'{a}zman, 1986; Dette et al., 2008; Yu, 2010a).  B\"{o}hning (1986) has given a proof of the monotonic convergence of Algorithm II, i.e., VEM.  Theoretical results concerning algorithms related to VEM and VDM can be found in Atwood (1976) and Wu (1978), for example. 

The monotonicity of Algorithm III is immediate since VDM, the nearest neighbor exchanges, and Algorithm I are all monotonic.  That Algorithm III converges is a consequence of this monotonicity and the global convergence nature of VDM. 

\begin{theorem}
\label{thm1} 
Assume the $n\times m$ matrix $X=(x_1, \ldots, x_n)^\top$ has full rank $m$.  Then Algorithm III converges monotonically starting from any $w^{(0)}\in \Omega_+$ where $\Omega_+=\{w\in \bar{\Omega}:\ \det M(w)>0\}$.
\end{theorem}
\begin{proof}
See the Appendix.
\end{proof}

\section{Numerical examples}
We illustrate the effectiveness of the cocktail algorithm by comparing it with Algorithms I and II for a few regression models.  VDM by itself is very slow and is excluded from the comparisons.  All algorithms are implemented in R, and the source code is available upon request from the author.  The main program contains fewer than 150 lines of code, showing that Algorithms I--III are indeed easy to implement.  We also consider general-purpose algorithms such as Nelder-Mead,  conjugate gradient (CG), and quasi-Newton (specifically, the Broyden-Fletcher-Goldfarb-Shanno, or BFGS method).  These are known to be powerful for solving various high-dimensional optimization problems.  However, they are not the most effective for the D-optimal design problem considered here. 

For Algorithms~I--III, both the number of iterations and the computer time (as measured by the R function system.time()) are reported.  An iteration of the cocktail algorithm is counted as one iteration each of VDM, NNE and MA.  It may seem that the  iteration count comparison would favor the cocktail algorithm unfairly.  Careful inspection, however, shows that the computing time per iteration for the cocktail algorithm is spent mainly by the VDM step, because the NNE and MA steps only work with design points that receive positive mass in the current iteration, and this set of support points is typically much fewer than $n$.  Consequently the computing costs per iteration are actually comparable for VEM (i.e., Algorithm II) and the cocktail algorithm.  At any rate, the reader is reminded to focus on the computing time comparisons. 

For VEM and the cocktail algorithm, the starting design $w^{(0)}$ is the uniform design over a set of approximately $2m$ randomly sampled support points.  This is intended to ensure that $\det M(w^{(0)})>0$ while keeping the number of support points small.  VEM tends to take more iterations if the initial design has more support points, since it can remove at most one bad support point per iteration.  It is observed that the cocktail algorithm is relatively insensitive to the initial number of support points.  The multiplicative algorithm is always started at the uniform design over all $n$ points, as it cannot afford to exclude any design point {\it a priori} (see, however, Harman and Pronzato 2007). 

We consider the design spaces
\begin{align*}
\mathcal{X}_1(n) &=\left\{x_i=(e^{-s_i},\, s_i e^{-s_i}, e^{-2s_i}, s_i e^{-2s_i})^\top,\ 1\leq i\leq n\right\},\\
\mathcal{X}_2(n) &=\left\{x_i=(1,\ s_i,\ s_i^2,\ s_i^3,\ s_i^4)^\top,\ 1\leq i\leq n\right\},\\
\mathcal{X}_3(n) &=\left\{x_i=(e^{-s_i},\, s_i e^{-s_i}, e^{-2s_i}, s_i e^{-2s_i}, e^{-3s_i}, s_i e^{-3s_i}, e^{-4s_i}, s_i e^{-4s_i})^\top,\ 1\leq i\leq n\right\}, 
\end{align*}
where $s_i=3i/n,\ i=1, \ldots, n$.  The space $\mathcal{X}_1(n)$ represents the linearization of a compartmental model (see, e.g., Atkinson et al. 1993, and Dette, Melas and Wong 2006) 
$$y|(s, \theta) \sim \theta_1 e^{-\theta_2 s} + \theta_3 e^{-\theta_4 s} +{\rm N}(0, \sigma^2)$$
at $\theta_2=1$ and $\theta_4=2$ (the underlying design variable is $s\in[0,3]$ on a grid of $n$ evenly spaced points).
The space $\mathcal{X}_3(n)$ is similar to $\mathcal{X}_1(n)$ but has a parameter of higher dimension.  We include polynomial regression as represented by $\mathcal{X}_2(n),$ although analytic results are well known in this case (see, e.g., Pukelsheim, 1993).  For $s_i=i/k,\ r_i=2i/k-1,\ i=1,\ldots, k,$ we also consider 
\begin{align*}
\mathcal{X}_4(k^2) &=\left\{x_{(i-1)k+j}=(1,\ r_i,\ r_i^2,\ s_j,\ r_i s_j)^\top,\ 1\leq i, j\leq k\right\}. 
\end{align*}
This last example represents a response surface with a nonlinear effect and an interaction. 

Each algorithm is stopped when either the convergence criterion (\ref{conv}) is met with $\epsilon =10^{-6}$, or the number of iterations exceeds 10000.  For large design spaces, some of the experiments are aborted because the algorithm under consideration takes too much time, especially compared with the cocktail algorithm.  We have also considered less stringent convergence criteria such as $\epsilon=10^{-5}$, and the results 
(omitted) are similar. 

As is evident from Tables 1--4, the cocktail algorithm is a substantial improvement over both the multiplicative algorithm and VEM.  Because of the random starting values, results for VEM and the cocktail algorithm vary from replication to replication; the qualitative comparison, however, remains the same.  The tables report the median computing time (iteration count) over three replications for VEM and the cocktail algorithm.  For $\mathcal{X}_1(n)$ and $\mathcal{X}_3(n)$, VEM is much faster than the multiplicative algorithm; the situation is less clear for $\mathcal{X}_2(n)$.  The cocktail algorithm improves upon the better of the two, often by large factors.  MA and VEM tend to take more iterations for larger $n$, i.e., when the design space becomes finer, although peculiar exceptions do exist (e.g., VEM in Table~2).  The cocktail algorithm seems insensitive to $n$ concerning the number of iterations, at least for the design spaces considered. 

We also consider Nelder-Mead, conjugate gradient (CG), and quasi-Newton algorithms, which are readily available via the R function optim().  Quasi-Newton here refers to the popular Broyden-Fletcher-Goldfarb-Shanno (BFGS) method, while  conjugate gradient uses Fletcher-Reeves updates.  These are tested on the same design spaces and compared with Algorithms~I--III.  To make the optimization problem unconstrained, we use the substitution $w_i=z_i^2/\sum_{j=1}^n z_j^2,\ i=1,\ldots, n,$ and operate on $z$ rather than $w$ (see Atkinson et al.\ 2007).  The starting value is $z_i\equiv 1$.  We use numerical derivatives for BFGS and conjugate gradient.  It should be noted that these general purpose algorithms are not guaranteed to find a global maximum.  In several cases, despite extensive tuning, we have been unable to obtain an output that satisfies our convergence criterion (\ref{conv}) with $\epsilon=10^{-6}$.  Nelder-Mead, for example, often stops at sub-optimal solutions; so do BFGS and conjugate gradient for $\mathcal{X}_3(n)$.  In other cases, and with moderate $n$, we record the computing time of BFGS and conjugate gradient in Tables~1, 2 and 4.  BFGS seems faster than conjugate gradient in these cases and is sometimes competitive with the better of VEM and MA (e.g., for $\mathcal{X}_2(20)$ or $\mathcal{X}_2(50)$).  However, it definitely takes more time than the cocktail algorithm.  We note that one must be cautious when making such quantitative comparisons between algorithms with very different structures, because details of implementation may affect the relative performance considerably.  Nevertheless, this limited experience makes us more confident in recommending the cocktail algorithm, which is simple and fast, and has a global convergence guarantee. 

\begin{table}
\caption{Computing time (in seconds) and number of iterations (in parentheses) for the multiplicative algorithm (MA), the vertex exchange method (VEM), and the cocktail algorithm, for design space $\mathcal{X}_1(n)$.  Also included is the computing time for conjugate gradient (CG) and quasi-Newton (BFGS) methods.}
\begin{center}
\begin{tabular}{l|rrrrr}
\hline                    & $n=20$       & $n=50$          & $n=100$       & $n=200$        & $n=500$       \\
\hline
CG                        & 14.5         & 111.3          & 1328.4         &                &               \\
BFGS                      & 8.82         & 39.8           & 293.4          &                &               \\
MA                        & 14.3 (4239)  & 63.7 (8015)    & 147+ (10000+)  & 307+ (10000+)  & 762+ (10000+) \\
VEM                       & 0.17 (58)    & 1.43 (241)     & 23.1 (2113)    & 206+ (10000+)  & 555+ (10000+) \\
cocktail                  & 0.07 (8)     & 0.11 (9)       & 0.25 (13)      & 0.36 (13)      & 0.96 (16)     \\
\hline 
\end{tabular}
\end{center}
\end{table}

\begin{table}
\caption{Computing time (in seconds) and number of iterations (in parentheses) for design space $\mathcal{X}_2(n)$.}
\begin{center}
\begin{tabular}{l|rrrr}
\hline                    & $n=20$      & $n=50$       & $n=100$     & $n=200$        \\
\hline
CG                        & 4.61        & 164.5        & 220.3       &                \\
BFGS                      & 1.83        & 14.7         & 116.1       &                \\
MA                        & 3.38 (947)  & 10.1 (1292)  & 76.3 (4105) & 427+ (10000+)  \\
VEM                       & 6.32 (1371) & 30.3 (4747)  & 4.04 (302)  & 252+ (10000+)  \\
cocktail                  & 0.31 (24)   & 0.65 (25)    & 0.21 (10)   & 0.63 (21)      \\
\hline 
\end{tabular}
\end{center}
\end{table}

\begin{table}
\caption{Computing time (in seconds) and number of iterations (in parentheses) for design space $\mathcal{X}_3(n)$.}
\begin{center}
\begin{tabular}{l|rrrr}
\hline                    & $n=20$      & $n=50$       & $n=100$       & $n=200$        \\
\hline
MA                        & 3.94 (609)  & 24.2 (2371)  & 44.7 (3016)   & 382+ (10000+)  \\
VEM                       & 0.80 (182)  & 10.7 (1291)  & 37.6 (3242)   & 127.2 (5324)  \\
cocktail                  & 0.72 (22)   & 1.56 (32)    & 1.34 (42)     & 1.21 (29)      \\
\hline 
\end{tabular}
\end{center}
\end{table}

\begin{table}
\caption{Computing time (in seconds) and number of iterations (in parentheses) for design space $\mathcal{X}_4(n)$.}
\begin{center}
\begin{tabular}{l|rrrr}
\hline                    & $n=20^2$      & $n=50^2$     & $n=100^2$   & $n=200^2$   \\
\hline
CG                        & 1545.9        &              &             &             \\
BFGS                      & 657.3         &              &             &             \\
MA                        & 25.2 (430)    & 993.8 (2302) &             &             \\
VEM                       & 8.01 (159)    & 195.8 (702)  & 94.6 (98)   &             \\
cocktail                  & 0.63 (13)     & 3.94 (14)    & 17.6 (14)   &  74.1 (16)  \\
\hline 
\end{tabular}
\end{center}
\end{table}

\section{Discussion}
Although we focus on D-optimal designs for linear models, the basic idea is not limited to 
either D-optimality or linear models.  The multiplicative algorithm can be more general and 
is known to be monotonic for a large class of optimality criteria (Silvey et al.\ 1978, Yu 2010a).
For vertex exchange strategies with optimality criteria other than D-optimality, we may not 
have a closed form solution for the maximizing step-length similar to (\ref{dstar}).  But such
one-dimensional maximization problems presumably can be handled by standard tools such as 
Newton's method.  The idea of nearest neighbor exchanges is generic.  Overall, although the 
implementation may not be as simple, there is no conceptual problem extending the cocktail 
algorithm to other optimality criteria or to nonlinear problems. 

The optimal design problem is closely related to several other statistical problems (Haines 1998) such as
mixture estimation (Lindsay 1983) and nonparametric estimation with censored data. 
There exists a large literature on efficient computation of the nonparametric MLE of the distribution function with censored data; see, for example, Wellner and Zhan (1997), Jongbloed (1998) and Wang (2008).  The cocktail algorithm can be extended to this case and is quite competitive; see 
Yu (2010c). 

\section*{Acknowledgments}
This work is partly supported by a CORCL special research grant from the University of California,
Irvine.  The author would like to thank Don Rubin, Xiao-Li Meng, and David van Dyk for introducing him to the field of statistical computing.  He is also grateful to Anatoly Zhigljavsky, Yong Wang, an associate editor, and two referees for their valuable comments. 

\appendix
\section*{Appendix: Proof of Theorem \ref{thm1}}
Let $w^{(t)}$ be a sequence generated by Algorithm III, and let $w^{(t_j)}$ be a convergent subsequence tending to some $w^*$.  Monotonicity and $w^{(0)}\in \Omega_+$ show that $w^{(t)}\in \Omega_+$ for all $t$.  Hence $w^*\in \Omega_+$.  Let $w^{(t+1/3)}$ be defined as in (\ref{cocktail}).  By passing through another subsequence if necessary, we may assume $w^{(t_j+1/3)}$ converges to some $\tilde{w}\in \Omega_+$.  Moreover, we may assume that the VDM steps
\begin{equation*}
w^{(t_j+1/3)}=VDM(w^{(t_j)})
\end{equation*}
are all performed with the same index $k=i_{max}$ as in (\ref{VDM}), since at least one of the $n$ indices will occur infinitely often. 

The mapping $w^{new}=VDM(k, w)$ is continuous on 
$$\{w\in \Omega_+:\ d(k, w)\geq d(i, w),\, 1\leq i\leq n\}.$$ 
By letting $j\to \infty$ in $w^{(t_j+1/3)}=VDM(w^{(t_j)}),$ we get $\tilde{w}=VDM(k, w^*).$  Because each step of VDM, NNE, or MA is monotonic, and $t_j+1\leq t_{j+1}$, we have 
$$\det M(w^{(t_j)})\leq \det M(w^{(t_j+1/3)}) \leq \det M(w^{(t_j+1)})\leq \det M(w^{(t_{j+1})}).$$
Letting $j\to\infty$ yields $\det M(w^*)=\det M(\tilde{w})$.  However, inspection shows that the mapping $VDM(k, w)$ strictly increases $\det M(w)$, unless $d(k, w)=m$, in which case $w=VDM(k, w)$.  Hence $\tilde{w}=w^*$ and $d(k, w^*)=m$. Since $k=i_{max}$, the general equivalence theorem implies that $w^*$ is a global maximizer of $\det M(w)$ on $\Omega_+$.  That is, all limit points of $w^{(t)}$ are D-optimal.

\end{document}